\documentclass[journal]{IEEEtran}

\usepackage{epsfig}
\usepackage{amssymb}
\usepackage{amsmath}
\usepackage{color}
\usepackage{caption}
\usepackage{graphicx}
\usepackage{float}
\usepackage{subcaption}

\newtheorem{theorem}{Theorem}

\newtheorem{definition}{Definition}
\newtheorem{lemma}{Lemma}
\newtheorem{corollary}{Corollary}

\newtheorem{remark}{Remark}

\newtheorem{assumption}{Assumption}

\newtheorem{proof}{Proof}
\newtheorem{proof of Theorem 1}{Proof of Theorem 1}

\begin{document}

\title{Online Distributed Zeroth-Order Optimization With Non-Zero-Mean Adverse Noises}

\author{Yanfu Qin and Kaihong Lu
\thanks{ }
\IEEEcompsocitemizethanks{
	\IEEEcompsocthanksitem Corresponding author: Kaihong Lu
	\IEEEcompsocthanksitem Y. Qin and K. Lu are with the College of Electrical Engineering and Automation, Shandong University of Science and Technology, Qingdao 266590, China. (e-mail: qinyan{\_}fu@163.com; khong{\_}lu@163.com) }
}
\maketitle

\begin{abstract}

In this paper, the problem of online distributed zeroth-order optimization subject to a set constraint is studied via a multi-agent network, where each agent can communicate with its immediate neighbors via a time-varying directed graph.
Different from the existing works on online distributed zeroth-order optimization, we consider the case where the estimate on the gradients are influenced by some non-zero-mean adverse noises.
To handle this problem, we propose a new online distributed zeroth-order mirror descent algorithm involving a kernel function-based estimator and a clipped strategy. Particularly, in the estimator, the kernel function-based strategy is provided to deal with the adverse noises, and eliminate the low-order terms in the Taylor expansions of the objective functions.
Furthermore, the performance of the presented algorithm is measured by employing the dynamic regrets, where the offline benchmarks are to find the optimal point at each time. Under the mild assumptions on the graph and the objective functions, we prove that if the variation in the optimal point sequence grows at a certain rate, then the high probability bound of the dynamic regrets increases sublinearly. 	
Finally, a simulation experiment is worked out to demonstrate the effectiveness of our theoretical results.


\end{abstract}


\IEEEpeerreviewmaketitle
\section{Introduction}
\IEEEPARstart{I}{n} online distributed optimization, the goal of agents is to cooperatively minimize the sum of objective functions in dynamic environments \cite{8015179}.
In recent years, online distributed optimization has been received increasing attention \cite{7172037}, \cite{7399359}, \cite{7479495}, \cite{8421588}, \cite{9013030}, \cite{8710288}, \cite{9851519}, \cite{10025380}, \cite{10664010}.
This is due to its wide applications in many areas such as Internet of things \cite{8501581}, smart grid \cite{9585464},  robot formation \cite{STOMBERG2023105579}.

An online algorithm should mimic the performance of its offline counterpart, and the gap between them is called the regret \cite{7172037}. In \cite{7172037}-\cite{9013030} the static regret, whose offline benchmark is to minimize the average of global objective functions of all time, is used to measure the performance of online distributed algorithms.
While in \cite{8710288}-\cite{10664010} the dynamic regret, whose offline benchmark is to  minimize the global objective function at each time, is used to capture the performance of online distributed algorithms.
Undoubtedly, the offline benchmark of the dynamic regrets is more stringent than that of the static ones.

All works in \cite{7172037}–\cite{10664010} assume that each agent can access the real gradient information of its objective function.
However, computing the real gradients usually takes expensive costs, even is impossible in some applications \cite{10480627}.
For the cases where real gradients of the objective functions are not available, the gradient can be estimated by using zeroth-order estimate methods. Accordingly, the optimization problems are called zeroth-order optimization \cite{7298442}.
Recently, online distributed zeroth-order optimization has been extensively studied.
For example, for online distributed zeroth-order optimization without constraints, a contextual learning algorithm based on the multi-point estimation is proposed in \cite{7103356}, and a quantized distributed algorithm based on the one-point estimation is proposed in \cite{YUAN2022110590}.
For the case with time-varying coupled inequality constraints, distributed primal-dual algorithms based on the one-point estimate strategy and the two-point estimate strategy are proposed in \cite{9222230}.
For online distributed zeroth-order optimization with long-term constraints, a distributed primal-dual algorithm based on the one-point estimation is proposed \cite{9349205}. Moreover, with the coupled inequality constraints considered, a modified saddle-point algorithm based on the two-point estimate strategy is proposed in \cite{9806334}.
For online distributed zeroth-order optimization with nonconvex and nonsmooth objective functions, a multi epoch distributed algorithm based on the two-point estimation is proposed in \cite{pmlr-v235-sahinoglu24a}.


The above study focuses on online distributed optimization problems without adverse noises.
Unfortunately, the estimator is usually influenced by adverse noise in practical applications.
For example, in the image classification problems, misclassification are made in deep neural networks due to the fact that the datasets are often polluted by adverse noises \cite{Tu_Ting2019}.
For distributed zeroth-order optimization with zero-mean noises, an distributed Kiefer-Wolfowitz stochastic approximation algorithm is proposed in \cite{8619044}.
For distributed zeroth-order optimization with sub-Gaussian noises, an distributed algorithm based on Gaussian process method is proposed in \cite{11016034}.
It is worth pointing out that, all the aforementioned results on online distributed zeroth-order optimization are achieved by using the Gaussian approximation. The estimated error between the real gradient and that of the objective function's Gaussian approximation is linear with a constant smoothness coefficient, which results in a large error bound and causes a bad convergence performance. To improve the convergence performance, developing new distributed zeroth-order estimate methods to reduce the estimate errors are desired.

In this paper,  the problem of online distributed zeroth-order optimization is study via a multi-agent system.
When making decisions, each agent only has to access the zeroth-order information of  its own objective function and the set constraint, and can exchange local state information with its immediate neighbors via a time-varying directed graph.
Different from \cite{7103356}-\cite{pmlr-v235-sahinoglu24a}, here adverse noises are considered in zeroth-order gradients. Moreover, we consider the case where the means of the noises considered in this paper are not zero, as opposed to the cases studied in the offline distributed optimization \cite{8619044}, \cite{11016034}.
To handle the problem, we propose an online distributed zeroth-order mirror descent algorithm based on the kernel function-based estimator and clipped strategy.
In the estimator, the kernel function of a noise following the uniform distribution is used to deal with the adverse noises, and eliminate the low-order terms in the Taylor expansions of the objective functions.
Using the kernel function-based strategy, the estimate errors scale with a high order term of the estimating coefficient.
Compared with those achieved by zeroth-order methods based on the Gaussian approximation in \cite{7103356}-\cite{pmlr-v235-sahinoglu24a}, \cite{11016034} the bounds of the estimate errors are much smaller.
Furthermore, dynamic regret is employed to measure the performance of the online algorithm.
Different from \cite{7103356}-\cite{pmlr-v235-sahinoglu24a}, \cite{8619044}, \cite{11016034} where the expectation bounds of the dynamic regrets are analyzed, we study the high probability bound of the dynamic regrets, which help ensure the effectiveness of running the online algorithms in a few rounds. We prove that if the digraph is uniformly strongly connected and if the increasing rate of the variation in the optimal value sequence is slower than $\mathcal{O}(T^{1+b})$,  then the high probability bound of the dynamic regrets increases sublinearly.

{\emph{\textbf{Notations}}}. Throughout this paper, $\mathbb{R},~\mathbb{R}^m$ denote the set of real numbers and the space of $m$-dimensional real column vectors, respectively.
$e_i$ denotes the unit vector whose $i$-th element is $1$ and all other elements are $0$.
$\lfloor x\rfloor$ denotes the largest integer less than $x$.
For any positive integer $T$, $\lceil T \rceil$ denotes the sequence $\{1,\cdots T\}$.
$\langle x,y\rangle$ denotes the inner product of vectors $x$ and $y$.
$\left[x\right]_k$ denotes the $k$-th element of the vector $x$.
$\left[A\right]_i$ denotes the $i$-th row of the matrix $A$.
$\mathbb{P}[\cdot]$ denotes the probability of a random event.

\section{Problem formulation}\label{se2}

\subsection{Graph theory}

Consider a time-varying directed graph $\mathcal{G}(t)=(\mathcal{V}, \mathcal{E}(t), A(t))$, where $\mathcal{V}=\left\lbrace 1,\cdotp\cdotp\cdotp,n \right\rbrace $ represents a set of vertices, $\mathcal{E}(t)$ represents a set of edges, and $A(t)=(a_{ij}(t))_{n\times n}$ represents a weight matrix. If $(j,i)\in\mathcal{E}(t)$ then $l \leq a_{ij}(t)\leq 1$ for some $0< l < 1$  and $a_{ij}(t)=0$ otherwise.
$\mathcal{N}_i(t)=\left\lbrace j|(j, i)\in\mathcal{E}(t)\right\rbrace\cup\left\lbrace i \right\rbrace $ is used to represent the neighbor set of $i$.
$\mathcal{G}(t)$ is strongly connected if there exists a directed path between each pair of nodes. For $\mathcal{G}(t)$, defined the edge set as $\mathcal{E}_U(k)=\bigcup_{t=kU}^{(k+1)U-1}\mathcal{E}(t)$ for some positive integer $U>1$. If $\mathcal{G}(t)$ with $\mathcal{E}_U(k)$ is strongly connected for any $t\geq0$, then $\mathcal{G}(t)$ is called a uniformly strongly connected graph.


\begin{assumption} \label{as1}
	$\mathcal{G}(t)$ is balanced and uniformly strongly connected and $A(t)$ is a doubly stochastic matrix.
\end{assumption}

\begin{lemma} \label{le1} \cite{4749425}
	Under Assumption \ref{as1}, for any $i,j\in \mathcal{V}$,
	\begin{equation}
		\begin{split}
			\| [A(t,s) ]_{ij}- \frac{1}{n}\| \leq \mathcal{C}\lambda^{t-s},~t\geq s\geq 0
		\end{split}
	\end{equation}
	where $A(t,s)=A(t)\cdots A(s)$, $\mathcal{C}=2\frac{1+l^{-(n-1)U}}{1-l^{(n-1)U}}$ and $\lambda=(1-l^{(n-1)U})^{\frac{1}{(n-1)U}}$.
\end{lemma}

\subsection{Distributed  optimization}
Consider a multi-agent system consisting of $n$ agents.
The agents communicate with their immediate neighbors via time-varying digraph $\mathcal{G}(t)$.
After the state $x_i(t)$ is selected from a set $\Omega\subseteq \mathbb{R}^m$, the information of objective function $f_i^t(\cdot)$ is revealed to agent $i$ at time $t\in\lceil T \rceil$, where $T$ is the time horizon.
The goal of agents is to cooperatively solve the following optimization problem:
\begin{equation}\label{1}
	\begin{split}
		&\min_{{x}\in \mathbb{R}^m} f^t(x),~f^t(x)=\frac{1}{n}\sum_{i=1}^n f_i^t(x)	\\
		&\textrm{subject to}~~{x} \in \Omega
	\end{split}
\end{equation}
where $f_i^t(\cdot): \mathbb{R}^m\rightarrow\mathbb{R}$. At iteration $t$, agent $i$ can only access the noises value of the objective function $f_i^t(\cdot)$ after a decision are made.

Some basic assumptions are made for the problem.
\begin{assumption} \label{as2} 	
	1) $\Omega$ is both convex and compact; \\
	2) $f_i^t(\cdot)$ is convex, differentiable, and time-varying.
\end{assumption}

Under Assumption \ref{as2}, it follows that there exist some positive constants $B,D,G$ such that
\begin{equation*}
	\begin{split}
		\|x-y\|\leq B,  ~\left\|f_i^t(x)\right\|\leq D, \left\|\nabla f_i^t(x)\right\|\leq G ~~\forall x, y\in\Omega.
	\end{split}
\end{equation*}

\begin{definition}(Hölder-type condition)
	The function $f(\cdot):\mathbb{R}^m\to\mathbb{R}$ satisfies a Hölder-type condition, when there exists a real number $H>0$ and a positive integer $\epsilon\geq2$, and $\ell=\lfloor\epsilon \rfloor$, for any $x,y\in\mathbb{R}^m$, such that
	\begin{equation} \label{Hölder}
		\begin{split}
			\big| f(x)- \sum_{0\leq|\rho|\leq\ell}\frac{\partial^\rho f(y)}{\rho!}(x-y)^\rho\big|\leq H\|x-y\|^\epsilon
		\end{split}
	\end{equation}
	where the multi-index $\rho=(\rho_1,\ldots,\rho_m)$ is the $m$-dimensional vector of nonnegative integers, $\partial^\rho=\partial_{1}^{\rho_1}\ldots\partial_{m}^{\rho_m}$ is the mixed partial derivative, $|\rho|=\rho_1+\ldots+\rho_m$,  $\rho!=\rho_1!\ldots\rho_m!$, and $(x-y)^\rho=[x-y]_{1}^{\rho_1}\ldots[x-y]_{m}^{\rho_m}$.
\end{definition}

A function that satisfies the Hölder-type condition is called the $\epsilon$-Hölder function.
Next, the assumption on the Hölder-type condition of the objective function is made, which is commonly used in the zeroth-order optimization problems \cite{pmlr-v130-liu21f}, \cite{ahookhosh2024high}, \cite{lin2025perseus}.

\begin{assumption} \label{as3}
	$f_i^t(\cdot)$ is an $\epsilon$-Hölder function.
\end{assumption}

Based on Assumption \ref{as3}, $f_i^t(\cdot)$ is twice differentiable. Together with the compactness of $\Omega$ in Assumption \ref{as2}, we know that $\nabla f_i^t(\cdot)$ is $L_0$-Lipschitz continuous, i.e., there exists a constant $L_0>0$ such that
\begin{equation} \label{L_0}
	\begin{split}
		\|\nabla f_i^t(x)-\nabla f_i^t(y)\| \leq L_0\|x-y\|~~\forall x, y\in\Omega.
	\end{split}
\end{equation}

In online optimization, the performance of algorithms should be measured by the regret.
Motivated by \cite{9184135}, \cite{XU2024111863}, we define the dynamic regrets as
\begin{equation}\label{regret}
	\begin{split}
		\mathcal{R}_i^d(T)=\sum_{t=1}^{T}\big(f^t(x_i(t))-f^t(x^*(t))\big).
	\end{split}
\end{equation}
An online algorithm performs well if dynamic regret (\ref{regret}) increase sublinearly, that is, $\lim_{T\to\infty}\frac{\mathcal{R}_i^d(T)}{T}=0$.
It is well known that using the dynamic regret causes the problem insolvable in the worst case where the objective functions change rather fast \cite{9462469}, \cite{9184135}, \cite{XU2024111863}.
Here we use the following deviation of the minimizer sequence to measure the difficulty
\begin{equation}\label{constraint}
	\begin{split}
		\Xi_T=\sum_{t=1}^{T}\|x^*(t+1)-x^*(t)\| .
	\end{split}
\end{equation}

\subsection{Algorithm design}

Since the real gradient is not available, the following kernel function-based estimator is used to estimate the zeroth-order gradient
\begin{equation} \label{step0}
	\begin{split}
		\left\{
		\begin{array}{l}
			g_{i,l}^t(x_i(t)) =\frac{f_i^t(x_i(t) + \gamma_t r_{i}(t) e_l)-f_i^{t}(x_i(t) - \gamma_t r_{i}(t) e_l)}{2\gamma_t}+\xi_{i,l}(t)	\\
			\widehat{\nabla}f_{i,l}^t(x_i(t))=g_{i,l}^t(x_i(t))K(r_i(t))
		\end{array}\right.
	\end{split}
\end{equation}
where $\gamma_t$ is a estimating parameter such that $\gamma_t>0$, $r_i(t)$ is a random perturbation following a uniform distribution on $[-1,1]$, $\xi_{i,l}(t)$ is an adverse noise caused by external interference,
$r_i(t)$ and $\xi_{i,l}(t)$ are independent for any $i\in\mathcal{V}$ and $l\in\{1,\cdots m\}$,
and $K(\cdot):[-1,1]\to\mathbb{R}$ is a kernel function satisfying $\int r K(r)dr=2$, $\int r^a K(r)dr=0$, $\kappa_\epsilon\equiv\int|r|^\epsilon |K(r)| dr<\infty$, $\kappa\equiv\int K^2(r)dr<\infty$, for $a=0,2,3,\ldots,\ell$, $\epsilon\geq2$, and $\ell=\lfloor \epsilon\rfloor$.

\begin{assumption} \label{as4} For any $i\in\mathcal{V}$, $l\in \{1,\cdots,m\}$,	
	$\mathbb{E}[(\xi_{i,l}(t))^2]\leq\sigma_{i,l}^2$.
\end{assumption}

Define $\sigma=\max_{i\in\mathcal{V},l\in \{1,\cdots,m\}}\sigma_{i,l}$.
Note that in Assumption \ref{as4}, we only assume that the variances of the adverse noises are bounded. The mean of the adverse noises is never required to be zero.
Now consider a differentiable and $\mu$-strongly convex function $\phi(\cdot):\Omega\rightarrow\mathbb{R}$.
The Bregman function associated with $\phi(\cdot)$ is defined as $D_{\phi}(x,y)=\phi(x)-\phi(y)-\langle \nabla\phi(y),x-y \rangle$.
By the strong convexity of $\phi(\cdot)$, we have $D_{\phi}(x,y)\geq\frac{1}{2}\|x-y\|^2$.

\begin{assumption} \label{as5} 	
	$D_{\phi}(\cdot,\cdot)$ is $L_1$-Lipschitz continuous with respect to its first argument and convex with respect to its second argument.	
\end{assumption}

To solve problem (\ref{1}),  we propose an online distributed zeroth-order mirror descent algorithm involving a kernel function-based estimator and a clipped strategy.
By running Algorithm 1, each agent makes decisions only using the zeroth-order gradient information of its own objective function in the past time and the state information received from its immediate neighbors. Thus, Algorithm 1 is online and distributed.

\begin{remark}
	In Algorithm 1, the design of dynamics (\ref{step0}) is motivated by the two-point gradient estimation method \cite{YUAN2022110590}, \cite{8629972} and the kernel function-based strategy \cite{dippon2003accelerated}, \cite{pmlr-v49-bach16}.
	And the design of dynamics (\ref{step3}) and (\ref{step4}) is inspired by the mirror descent algorithm in \cite{8015179}, \cite{9416872}, \cite{pmlr-v119-eshraghi20a}.
	Here the kernel function-based strategy is used to deal with the adverse noises and eliminate the low-order terms in the Taylor expansion of the objective function.
	Due to the influence of the adverse noises, the zeroth-order gradients $\widehat{\nabla}f_i^t(x_i(t))$ achieved by (\ref{step0}) follow a heavy-tailed distribution, which has heavier tails than the exponential distribution and therefore often appears extreme values.
	Motivated by \cite{NEURIPS2020_abd1c782}, clipped strategy (\ref{step2}) is employed to deal with the extreme values.

\end{remark}

In this paper, we are committed to studying the high probability bound of dynamic regret (\ref{regret}) under Algorithm 1.
\begin{definition}(High probability bound)
	Given $h(\cdot):\mathbb{R}\to\mathbb{R}$, if $\mathcal{R}_i^d(T)\leq\mathcal{O}(h(T)\ln\frac{1}{\delta})$ with probability at least $1-\delta$ for any $\delta\in(0,1)$, then $\mathcal{R}_i^d(T)$ is called to have a high probability bound.	
\end{definition}

\noindent
\rule[0\baselineskip]{8.9cm}{1pt}
\emph{Algorithm 1:  online distributed zeroth-order mirror descent}\\
\rule[0.53\baselineskip]{8.9cm}{1pt}
\textbf{Initialization:} Set the initial value as $x_i(1)\in\Omega$.	\\
\textbf{Iteration:}	 At each iteration time $t=1,2,\ldots$ and for any $i\in\mathcal{V}$, each agent $i$ updates variables using the following rules.

\textbullet~ The zeroth-order gradient $\widehat{\nabla}f_i^t(x_i(t))$ is computed by (\ref{step0}).

\textbullet~ Compute the clipped gradient $\widetilde{\nabla}f_i^t(x_i(t))$ as follows
\begin{equation}\label{step2}
	\begin{split}
		\widetilde{\nabla} f_i^t(x_i(t)) = \min\big\{1,~\frac{\alpha_t}{\| \widehat{\nabla}f_i^t(x_i(t))\|}\big\}\widehat{\nabla}f_i^t(x_i(t))
	\end{split}
\end{equation}
where $\widehat{\nabla}f_i^t\triangleq [\widehat{\nabla}f_{i,1}^t,\cdots,\widehat{\nabla}f_{i,m}^t]^{\top}$ and $\alpha_t$ is the clipping parameter satisfying $\alpha_t\geq2G$.

\textbullet~ Update the value of $y_i(t)$ as follows
\begin{equation}\label{step3}
	\begin{split}
		y_i(t)=\sum_{j\in\mathcal{N}_i}a_{ij}(t){x}_{j}(t).
	\end{split}
\end{equation}

\textbullet~ Update the value of $x_i(t+1)$ as follows
\begin{equation}\label{step4}
	\begin{split}
		{x}_{i}(t+1)=\underset{x \in \Omega}{\operatorname{argmin}}\big\{\beta_t\langle x, \widetilde{\nabla} f_i^t(x_i(t)) \rangle+ \mathcal{D}_{\phi}(x,y_i(t))\big\}
	\end{split}
\end{equation}
where $\beta_t$ is the non-increasing step-size satisfying $0<\beta_t<1$.
\rule[0.3\baselineskip]{8.9cm}{1pt}

\section{Main results}\label{se3}

In this section, we will provide our main result and its proof in detail. Let us start by presenting our main result in the following theorem.

\begin{theorem}	\label{th1}
	Under Assumptions \ref{as1}-\ref{as5}, by Algorithm 1, for any $i\in\mathcal{V}$ and $\delta\in(0,1)$, with probability at least $1-\delta$
	\begin{equation} \label{th11}
		\begin{aligned}
			\mathcal{R}_i^d(T)&\leq\mathcal{O}\big(\sum_{t=1}^{T}(\alpha_t^2\beta_t +\frac{\alpha_t^2}{\sqrt{T}} +\gamma_t^{\epsilon-1} +\lambda^{t-1}) +\frac{1+\Xi_T}{\beta_{T+1}} 	\\
			& +\sum_{t=1}^{T}\sum_{s=1}^{t-1}\alpha_{s}\beta_{s}\lambda^{t-1-s}+\sqrt{T}\ln\frac{1}{\delta}	+ \sum_{t=1}^{T}\frac{1}{\alpha_t}(\gamma_t^2+1)\big)
		\end{aligned}
	\end{equation}
	where $\Xi_T$ is defined in (\ref{constraint}).
\end{theorem}

\begin{corollary}
	Under Assumptions \ref{as1}-\ref{as5}, if $\alpha_t=t^a+2G$, $\beta_t=t^b$, $\gamma_t=t^c$ for some $0<a<\frac{1}{2}$, $-1<b<-2a$, $c<0$, then by Algorithm 1, for any $i\in\mathcal{V}$ and $\delta\in(0,1)$, with probability at least $1-\delta$
	\begin{equation} \label{coro}
		\begin{split}
			\mathcal{R}_i^d(T)&\leq\mathcal{O}\big(T^{1+2a+b} +T^{\frac{1}{2}+a} +T^{1+(\epsilon-1)c} +(1+\Xi_T)T^{-b} \\
			&~~~ +\sqrt{T}\ln\frac{1}{\delta} +T^{1-a+2c} \big) .
		\end{split}
	\end{equation}
\end{corollary}

From Corollary 1, the sublinearity of the bound in (\ref{coro}) is influenced by term $\sqrt{T}\ln\frac{1}{\delta}$. Note that the value of $\ln\frac{1}{\delta}$ increases slowly as the value of failure probability $\delta$ decreases.  The sublinearity of term $\ln\frac{1}{\delta}$ with a probability close to $100\%$ can be ensured \cite{10295561}.
Moreover, the sublinearity of the bound in (\ref{coro}) is also influenced by $\Xi_T $. If $\Xi_T $ is sublinear with $T^{1+b}$, i.e., $\lim_{T\to\infty}\frac{\Xi_T}{T^{1+b}}=0$, then $\mathcal{R}_i^d(T)$ has a high probability bound of sublinear.
This is natural since even using the real gradient information \cite{8710288}-\cite{10664010}, the problem is insolvable in worst cases when the minimizers change rather fast.

Before proving Theorem \ref{th1}, some necessary lemmas need to be established.
First, the error bound between the zeroth-order gradient and the real gradient is analyzed.
\begin{lemma}	\label{esti-true}
	Under Assumption \ref{as3}, by Algorithm 1, for any $i\in\mathcal{V}$
	\begin{equation}	\label{esti-true1}
		\begin{split}
			&\|\mathbb{E}[\widehat{\nabla}f_i^t(x_i(t))|\mathcal{F}_i^{t}] -\nabla f_i^t(x_i(t))\| \leq m\kappa_\epsilon H\gamma_t^{\epsilon-1}
		\end{split}
	\end{equation}
	where $\mathcal{F}_i^t=\sigma\left(x_{i}(s), r_{i}(s), \xi_{i,l}(s): s < t\right)$ is the filtration representing all known random information before time $t$, and $\kappa_\epsilon$ is defined in (\ref{step0}).
\end{lemma}

\begin{proof}
	For $f_i^t(x_i(t)+\gamma_t r_i(t) e_l)$, by Taylor expansion, we have
	\begin{equation*}		\label{esti-true2}
		\begin{split}
			&f_i^t(x_i(t)+\gamma_t r_i(t) e_l)	\\
			&=f_i^t(x_i(t))+\langle\nabla f_i^t(x_i(t)),\gamma_t r_i(t) e_l \rangle	\\
			&~+\sum_{2\leq|\rho|\leq\ell}\frac{\partial^\rho f_i^t(x_i(t))}{\rho!}(\gamma_t r_i(t) e_l)^\rho +R_{\epsilon}(\gamma_t r_i(t) e_l).	\\
		\end{split}
	\end{equation*}
	where the $R_{\epsilon}(\gamma_t r_i(t) e_l)$ is the high-order term.
	Then, for any $l\in \{1,\cdots,m\}$, one has
	\begin{equation}		\label{esti-true3}
		\begin{aligned}
			&\frac{f_i^t(x_i(t)+\gamma_t r_i(t) e_l)-f_i^t(x_i(t)-\gamma_t r_i(t) e_l)}{2\gamma_t}	\\
			&=\nabla_l f_i^t(x_i(t))r_i(t)+\sum_{2\leq|\rho|\leq\ell, |\rho| odd}\frac{\partial^\rho f_i^t(x_i(t))}{\gamma_t\rho!}(\gamma_t r_i(t) e_l)^\rho	\\
			&~+\frac{R_{\epsilon}(\gamma_t r_i(t) e_l)-R_{\epsilon}(-\gamma_t r_i(t) e_l)}{2\gamma_t}.
		\end{aligned}
	\end{equation}
	Combining (\ref{step0}) and (\ref{esti-true3}) results in that
	\begin{equation}	\label{esti-true4}
		\begin{split}
			&\big|\mathbb{E}[g_{i,l}^t(x_i(t)) K(r_i(t))|\mathcal{F}_i^{t}] -\nabla_l f_i^t(x_i(t)) \big|		\\
			&=\big|\mathbb{E}[\frac{R_{\epsilon}(\gamma_t r_i(t) e_l)-R_{\epsilon}(-\gamma_t r_i(t) e_l)}{2\gamma_t}K(r_i(t))|\mathcal{F}_i^{t}]\big|		\\
			&\leq \kappa_\epsilon H\gamma_t^{\epsilon-1}
		\end{split}
	\end{equation}
	where the first inequality results from (\ref{Hölder}).  Inequality (\ref{esti-true4}) immediately implies (\ref{esti-true1}).
\end{proof}

\begin{remark}
	In fact, the parameter $\gamma_t$ in (\ref{esti-true1}) plays a similar role as the smoothness coefficient of Gaussian approximation in guaranteeing the estimate error. The estimated error between the real gradient and that of the objective function's
	Gaussian approximation is linear with the smoothness coefficient \cite{YUAN2022110590}, \cite{9222230}, that is, the estimated error bound is $\mathcal{O}(\gamma_t)$.  Note that if $\gamma_t$ decays, $\gamma_t^{\epsilon-1}$ decays much faster because $\epsilon$ can be larger than $2$. 	
	More importantly, the bound of Lemma \ref{esti-true} may be $0$.
	For example, let $f(x)=x^3$, then a Taylor expansion of the function at point $x$ gives  $\frac{f(x+ry)-f(x-ry)}{2y} =3x^2r+y^2r^3$ for some $x,y,r\in\Omega$. Then, we have $\mathbb{E}[(\frac{f(x+ry)-f(x-ry)}{2y}+\xi)K(r)]=3x^2$, where $\xi$ is the adverse noises.	
	Ultimately, this implies that $\|\mathbb{E}[\widehat{\nabla}f(x)] -\nabla f(x)\| =0$, where $\mathbb{E}[\widehat{\nabla}f(x)]$ is the zeroth-order gradient.
\end{remark}

In the following lemma, we analyze the bound of $\|\widehat{\nabla}f_i^t(x_i(t)) -\nabla f_i^t(x_i(t))\|^2$.
\begin{lemma}	\label{le-hi}
	Under Assumption \ref{as2}-\ref{as4}, by Algorithm 1, for any $i\in\mathcal{V}$
	\begin{equation}	\label{le-hi1}
		\begin{split}
			&\mathbb{E}[\|\widehat{\nabla}f_i^t(x_i(t)) -\nabla f_i^t(x_i(t))\|^2|\mathcal{F}_i^{t}]	\\
			&\leq6m\kappa L_0^2\gamma_t^2 +4m\kappa\sigma^2+ 2G^2(6m\kappa +1)
		\end{split}
	\end{equation}
	where $L_0$ is defined in (\ref{L_0}) and $\kappa$ is defined in (\ref{step0}).
\end{lemma}

\begin{proof}
	For any $l\in \{1,\cdots,m\}$, we have
	\begin{equation}	\label{le-hi3}
		\begin{aligned}
			&\big(f_i^t(x_i(t)+\gamma_t r_i(t) e_l)-f_i^t(x_i(t) - \gamma_t r_i(t) e_l)\big)^2 \\
			&\leq3\big(f_i^t(x_i(t)+\gamma_t r_i(t) e_l)-f_i^t(x_i(t))	\\
			&~~-\langle\nabla f_i^t(x_i(t)), \gamma_t r_i(t) e_l\rangle \big)^2	\\
			&~+3\big(f_i^t(x_i(t)-\gamma_t r_i(t) e_l)-f_i^t(x_i(t))  \\
			&~~-\langle\nabla f_i^t(x_i(t)), -\gamma_t r_i(t) e_l\rangle \big)^2	+12\langle\nabla f_i^t(x_i(t)), \gamma_t r_i(t) e_l\rangle^2	\\
			&\leq3\big(\langle \nabla f_i^t(x_i(t)+\gamma_t r_i(t) e_l),\gamma_t r_i(t) e_l\rangle	\\
			&~~-\langle\nabla f_i^t(x_i(t)), \gamma_t r_i(t) e_l\rangle \big)^2	\\
			&~+3\big(\langle \nabla f_i^t(x_i(t)-\gamma_t r_i(t) e_l), -\gamma_t r_i(t) e_l)\rangle  \\
			&~~-\langle\nabla f_i^t(x_i(t)), -\gamma_t r_i(t) e_l\rangle \big)^2	+12\langle\nabla f_i^t(x_i(t)), \gamma_t r_i(t) e_l\rangle^2	\\
			&\leq6L_0^2\|\gamma_t r_i(t) e_l\|^4 +12\langle\nabla f_i^t(x_i(t)), \gamma_t r_i(t) e_l\rangle^2
		\end{aligned}
	\end{equation}
	where the second inequality holds by using the convexity of $\nabla f_i^t(\cdot)$ and the third one is true due to Assumption \ref{as3}.
	Note that
	\begin{equation}	\label{le-hifinal}
		\begin{split}
			&\mathbb{E}[\|\widehat{\nabla}f_i^t(x_i(t)) -\nabla f_i^t(x_i(t))\|^2|\mathcal{F}_i^{t}] \\
			&\leq2\mathbb{E}[\|\widehat{\nabla}f_i^t(x_i(t)) \|^2|\mathcal{F}_i^{t}] +2\|\nabla f_i^t(x_i(t))\|^2\\
			&=2\mathbb{E}[\|g_{i}^t(x_i(t))\|^2K^2(r_i(t))|\mathcal{F}_i^{t}] +2\|\nabla f_i^t(x_i(t))\|^2	\\
			&\leq\frac{m}{\gamma_t^2}\mathbb{E}[(f_i^t(x_i(t)+h_t r_i(t) e_j)-f(x_i(t)-h_t r_i(t) e_j))^2	\\
			&~~~K^2(r_i(t))|\mathcal{F}_i^{t}] \\
			&~+4m\mathbb{E}[\xi_{i,l}^2(t) K^2(r_i(t))|\mathcal{F}_i^{t}]+2\|\nabla f_i^t(x_i(t))\|^2	\\
			&\leq\frac{6mL_0^2}{\gamma_t^2}\mathbb{E}[\|\gamma_t r_i(t) e_l\|^4 K^2(r_i(t))|\mathcal{F}_i^{t}]	\\
			&~~+4m\sigma^2\mathbb{E}[ K^2(r_i(t))|\mathcal{F}_i^{t}] +2\|\nabla f_i^t(x_i(t))\|^2	\\
			&~+\frac{12m}{\gamma_t^2}\mathbb{E}[\langle \nabla f_i^t(x_i(t)),\gamma_t r_i(t) e_l\rangle^2 K^2(r_i(t))|\mathcal{F}_i^{t}]	\\
			&\leq6m\kappa L_0^2\gamma_t^2 +4m\kappa\sigma^2+ 2G^2(6m\kappa +1)
		\end{split}
	\end{equation}
	where the fourth inequality results from the fact that $\int r^2K^2(r)dr\leq\int K^2(r)dr\equiv\kappa$.
\end{proof}

Next, the high probability bound on the difference between the real gradient and the clipped gradient is presented.
\begin{lemma}\label{le4}
	Under Assumption \ref{as2} and \ref{as3}, for $\delta\in(0,1)$, with probability at least $1-\delta$
	\begin{equation} \label{le4-1}
		\begin{aligned}
			 &\sum_{t=1}^{T}\langle \nabla f_i^t(x_i(t))-  \widetilde{\nabla}f_i^t(x_i(t)), y_i(t)-x^*(t) \rangle	\\
			 &\leq\frac{2B^2}{\sqrt{T}}\sum_{t=1}^{T}\alpha_t^2 + \sqrt{T}\ln\frac{1}{\delta} + m\kappa_\epsilon BH\sum_{t=1}^{T}\gamma_t^{\epsilon-1}\\
			 &~+B\sum_{t=1}^{T}\frac{4}{\alpha_t}\big(6m\kappa L_0^2\gamma_t^2 +4m\kappa\sigma^2+ 12m\kappa G^2 +2G^2)\big).
		\end{aligned}
	\end{equation}
\end{lemma}

\begin{proof}
	According to the inequality $\exp(a)\leq\exp(a^2)+a$ for any $a\in\mathbb{R}$, there holds
	\begin{equation}	\label{le6-1}
		\begin{aligned}
			&\exp\big(\frac{1}{\sqrt{T}}\langle \mathbb{E}[\widetilde{\nabla}f_i^t(x_i(t))|\mathcal{F}_i^{t}] -\widetilde{\nabla}f_i^t(x_i(t)), y_i(t)-x^*(t) \rangle \big)	\\
			&\leq\exp\big(\frac{1}{T}\langle \mathbb{E}[\widetilde{\nabla}f_i^t(x_i(t))|\mathcal{F}_i^{t}] -\widetilde{\nabla}f_i^t(x_i(t)),	\\
			&~~y_i(t)-x^*(t) \rangle^2 		\big)		\\
			&~+\frac{1}{\sqrt{T}}\langle \mathbb{E}[\widetilde{\nabla}f_i^t(x_i(t))|\mathcal{F}_i^{t}] -\widetilde{\nabla}f_i^t(x_i(t)), y_i(t)-x^*(t) \rangle .	\\
		\end{aligned}
	\end{equation}
	Taking expectations on both sides of (\ref{le6-1}) yields
	\begin{equation}\label{le6-2}
		\begin{aligned}
			&\mathbb{E}\big[\exp\big(\frac{1}{\sqrt{T}}\langle \mathbb{E}[\widetilde{\nabla}f_i^t(x_i(t))|\mathcal{F}_i^{t}] -\widetilde{\nabla}f_i^t(x_i(t)), \\
			&~~y_i(t)-x^*(t) \rangle \big)|\mathcal{F}_i^{t}\big]	\\
			&\leq\mathbb{E}\big[\exp\big(\frac{1}{T}\langle \mathbb{E}[\widetilde{\nabla}f_i^t(x_i(t))|\mathcal{F}_i^{t}] -\widetilde{\nabla}f_i^t(x_i(t)),	\\
			&~~y_i(t)-x^*(t) \rangle^2 		\big)|\mathcal{F}_i^{t}\big]		\\
			&\leq\mathbb{E}\big[\exp\big(\frac{B^2}{T}( \| \mathbb{E}[\widetilde{\nabla}f_i^t(x_i(t))|\mathcal{F}_i^{t}] \|^2 \\
			&~~+ \| \widetilde{\nabla}f_i^t(x_i(t)) \|^2 )  \big)|\mathcal{F}_i^{t}\big]		\\
			&\leq\exp\left(\frac{2B^2}{T}\alpha_t^2\right)
		\end{aligned}
	\end{equation}
	where the second inequality holds by using the Cauchy-Schwarz inequality and Assumption \ref{as2}, and the last one results from $\| \widetilde{\nabla}f_i^t(\cdot) \|\leq\alpha_t$. Furthermore, we let
	\begin{equation*}
		\begin{split}
			\varphi(t)=\frac{1}{\sqrt{T}}\langle \mathbb{E}[\widetilde{\nabla}f_i^t(x_i(t))|\mathcal{F}_i^{t}] -\widetilde{\nabla}f_i^t(x_i(t)), y_i(t)-x^*(t) \rangle
		\end{split}
	\end{equation*}
	and consider the dynamics $\psi(t+1)=\exp\left(\varphi(t)-\frac{2B^2}{T}\alpha_t^2\right) \psi(t)$
	with $\psi(1)=1$. It is easy to verify that $ \psi(t+1)=\exp(\sum_{k=1}^{t}(\varphi(k)-\frac{2B^2}{T}\alpha_t^2))$. It follows from (\ref{le6-2}) that $\mathbb{E}[\psi(t+1)]\leq\mathbb{E}[\psi(t)]$. Taking the total expectation results in that
	\begin{equation*}
		\begin{split}
			\mathbb{E}[\psi(t+1)]\leq\mathbb{E}[\psi(t)]\leq\cdots\leq\mathbb{E}[\psi(1)]=1.
		\end{split}
	\end{equation*}
	Therefore, for any $Q\geq0$, there is
	\begin{equation*}
		\begin{split}
			&\mathbb{P}\big[\sum_{t=1}^{T}\big(\varphi(t)-\frac{2B^2}{T}\alpha_t^2\big)\geq Q\big]		\\
			&=\mathbb{P}\big[\exp\big(\sum_{t=1}^{T}\big(\varphi(t)-\frac{2B^2}{T}\alpha_t^2\big)\big)\geq\exp(Q)\big]		\\
			&\leq\frac{\mathbb{E}(\psi(T+1))}{\exp(Q)}		\\
			&\leq\exp(-Q)
		\end{split}
	\end{equation*}
	where the first inequality results from the Markov's inequality. According to the arbitrariness of $Q$, letting $Q=\ln\frac{1}{\delta}$ yields that for any $i\in\mathcal{V}$ and $\delta\in(0,1)$, with probability at least $1-\delta$,
	\begin{equation}	\label{le6-3}
		\begin{split}
			&\sum_{t=1}^{T}\frac{1}{\sqrt{T}}\langle \mathbb{E}[\widetilde{\nabla}f_i^t(x_i(t))|\mathcal{F}_i^{t}] -\widetilde{\nabla}f_i^t(x_i(t)), y_i(t)-x^*(t) \rangle	\\
			&\leq\frac{2B^2}{T}\sum_{t=1}^{T}\alpha_t^2 + \ln\frac{1}{\delta}.
		\end{split}
	\end{equation}
	By the fact that $\|\nabla f_i^t(\cdot)\|\leq G\leq\frac{\alpha_t}{2}$, we have
	\begin{equation*}\label{le7-2}
		\begin{split}
			&\| \widehat{\nabla}f_i^t(x_i(t)) \| 	\\
			&\leq\| \widehat{\nabla}f_i^t(x_i(t))-\nabla f_i^t(x(t)) \| + \|\nabla f_i^t(x_i(t)) \|		\\
			&\leq\| \widehat{\nabla}f_i^t(x_i(t))-\nabla f_i^t(x_i(t)) \| + \frac{\alpha_t}{2}	.
		\end{split}
	\end{equation*}
	Indicator functions $\omega_t$ and $\varpi_t$ are defined respectively as follows
			$\omega_t = 1\big\{\|\widehat{\nabla}f_i^t(x_i(t))\|\geq\alpha_t\big\}$
	and
		$	\varpi_t = 1\big\{ \| \widehat{\nabla}f(x(t))-\nabla f(x(t))\|>\frac{\alpha_t}{2} \big\} $.
	From the definitions of $\omega_t$ and $\varpi_t$, it follows that $\omega_t\leq\varpi_t$.	
	By (\ref{step2}) and $\omega_t$, we have
	\begin{equation*}
		\begin{split}
			&\widetilde{\nabla}f_i^t(x_i(t))	\\
			&=\frac{\alpha_t}{\| \widehat{\nabla}f_i^t(x_i(t))\|} \widehat{\nabla}f_i^t(x_i(t))\omega_t + \widehat{\nabla}f_i^t(x_i(t))(1-\omega_t)	\\
			&=\big(\frac{\alpha_t}{\| \widehat{\nabla}f_i^t(x_i(t)) \|} -1\big) \widehat{\nabla}f_i^t(x_i(t))\omega_t + \widehat{\nabla}f_i^t(x_i(t)).
		\end{split}
	\end{equation*}
	Hence
	\begin{equation*}\label{le7-3}
		\begin{split}
			&\|\nabla f_i^t(x_i(t))-\mathbb{E}[\widetilde{\nabla}f_i^t(x_i(t))|\mathcal{F}_i^{t}] \|		\\
			&\leq\|\nabla f_i^t(x_i(t))-\mathbb{E}[\widehat{\nabla}f_i^t(x_i(t))|\mathcal{F}_i^{t}] \|	\\
			&~+\| \mathbb{E}\big[\big(\frac{\alpha_t}{\| \widehat{\nabla}f_i^t(x_i(t)) \|}-1 \big) \widehat{\nabla}f_i^t(x_i(t))\omega_t |\mathcal{F}_i^{t}\big]  \|		\\
			&\leq\mathbb{E}\big[ \|\widehat{\nabla}f_i^t(x_i(t)) \|\big|1-\frac{\alpha_t}{\| \widehat{\nabla}f_i^t(x_i(t)) \|}\big|\omega_t |\mathcal{F}_i^{t}\big] +m\kappa_\epsilon H\gamma_t^{\epsilon-1}	\\
			&\leq\mathbb{E}[\|\widehat{\nabla}f_i^t(x_i(t))\|\omega_t|\mathcal{F}_i^{t}]	 +m\kappa_\epsilon H\gamma_t^{\epsilon-1}	\\
			&\leq\mathbb{E}[\| \widehat{\nabla}f_i^t(x_i(t))\|\varpi_t|\mathcal{F}_i^{t}]	+m\kappa_\epsilon H\gamma_t^{\epsilon-1}	\\		
			&\leq\mathbb{E}[\| \widehat{\nabla}f_i^t(x_i(t))-\nabla f_i^t(x_i(t))\|\varpi_t|\mathcal{F}_i^{t}] 	\\	
			&~+ \mathbb{E}[\|\nabla f_i^t(x_i(t))\|\varpi_t|\mathcal{F}_i^{t}]  +m\kappa_\epsilon H\gamma_t^{\epsilon-1}	\\
			&\leq\mathbb{E}[\|\widehat{\nabla}f_i^t(x_i(t))-\nabla f_i^t(x_i(t))\|^2|\mathcal{F}_i^{t}]^{\frac{1}{2}} \mathbb{E}[\varpi_t^{2}]^{\frac{1}{2}} 	\\
			&~+\| \nabla f_i^t(x_i(t))\|\mathbb{E}[\varpi_t|\mathcal{F}_i^{t}] +m\kappa_\epsilon H\gamma_t^{\epsilon-1}	\\						
			&\leq\mathbb{E}[\|\widehat{\nabla}f_i^t(x_i(t))-\nabla f_i^t(x_i(t))\|^2|\mathcal{F}_i^{t}]^{\frac{1}{2}}\mathbb{E}[\varpi_t|\mathcal{F}_i^{t}]^{\frac{1}{2}} \\
			&~+ \frac{\alpha_t}{2}\mathbb{E}[\varpi_t|\mathcal{F}_i^{t}] +m\kappa_\epsilon H\gamma_t^{\epsilon-1}
		\end{split}
	\end{equation*}
	where the second inequality holds by using the Jensen's inequality and Lemma \ref{esti-true}, and the sixth one results from Hölder's inequality. Note that
	\begin{equation*}	\label{le7-4}
		\begin{split}
			\mathbb{E}[\varpi_t|\mathcal{F}_i^{t}] &= \mathbb{P}[\| \widehat{\nabla}f_i^t(x_i(t))-\nabla f_i^t(x_i(t)) \|\geq\frac{\alpha_t}{2}]	\\
			&\leq \frac{\mathbb{E}[\|\widehat{\nabla}f_i^t(x_i(t))-\nabla f_i^t(x_i(t))\|^2|\mathcal{F}_i^{t} ]}{(\alpha_t/2)^2}	
		\end{split}
	\end{equation*}
	where the first inequality holds by using the Markov's inequality.
	Moreover,
	\begin{equation}	\label{le7-5}
		\begin{split}
			&\sum_{t=1}^{T}\langle \nabla f_i^t(x_i(t))-\mathbb{E}[\widetilde{\nabla}f_i^t(x_i(t))|\mathcal{F}_i^{t}], y_i(t)-x^*(t) \rangle		\\
			&\leq B\sum_{t=1}^{T}\|\nabla f_i^t(x_i(t))-\mathbb{E}[\widetilde{\nabla}f_i^t(x_i(t))|\mathcal{F}_i^{t}] \|	\\
			&\leq B\sum_{t=1}^{T}\big(\frac{4}{\alpha_t}\mathbb{E}[\|\widehat{\nabla}f_i^t(x_i(t))-\nabla f_i^t(x_i(t))\|^2 |\mathcal{F}_i^{t}] 	\\
			&~~~+m\kappa_\epsilon H\gamma_t^{\epsilon-1}\big).
		\end{split}
	\end{equation}
	Combining (\ref{le-hi1}), (\ref{le6-3}), and (\ref{le7-5}) implies (\ref{le4-1}).
\end{proof}

In the following lemma, the consensus error bound is presented.
\begin{lemma} \label{le2}
	Under Assumption \ref{as1}, for any $i\in\mathcal{V}$,
	\begin{equation} \label{le3-01}
		\begin{split}
			\| x_i(t+1)-\bar{x}(t+1) \| \leq \theta_1\lambda^t +\theta_2\sum_{s=1}^{t}\alpha_{s}\beta_{s}\lambda^{t-s}
		\end{split}
	\end{equation}
	where $\bar{x}(t)=\frac{1}{n}\sum_{i=1}^{n}x_i(t)$, $\theta_1=\frac{\sqrt{nm}\mathcal{C}}{\lambda}\|x(1)\|$, and $\theta_2=\frac{\sqrt{nm}\mathcal{C}}{\mu\lambda}$.
\end{lemma}

\begin{proof}
	By (\ref{step4}), for any $x\in\Omega$,
	\begin{equation}	\label{le2-1}
		\begin{split}
			&\langle \beta_t \widetilde{\nabla} f_i^t(x_i(t)) +\nabla\phi(x_i(t+1))- \nabla\phi(y_i(t)),\\
			 &~~x-x_i(t+1)\rangle \geq 0.
		\end{split}
	\end{equation}
	Letting $x=y_i(t)$, we have
	\begin{equation}	\label{le2-2}
		\begin{split}
			&\langle \beta_t\widetilde{\nabla} f_i^t(x_i(t)), y_i(t)-x_i(t+1) \rangle\\
			&\geq\langle \nabla\phi(y_i(t))-\nabla\phi(x_i(t+1)), y_i(t)-x_i(t+1) \rangle\\
			&\geq\mu\|y_i(t)-x_i(t+1)\|^2
		\end{split}
	\end{equation}
	where the second inequality holds due to the $\mu$-strongly convexity of $\phi(\cdot)$. Applying the Cauchy–Schwarz inequality to (\ref{le2-2}) yields
	\begin{equation}	\label{le2-3}
		\begin{split}
			\|y_i(t)-x_i(t+1)\|\leq \frac{\beta_t}{\mu} \|\widetilde{\nabla} f_i^t(x_i(t))\|	\leq \frac{\alpha_t\beta_t}{\mu}	\\
		\end{split}
	\end{equation}
	where the second inequality is true due to the fact that $\|\widetilde{\nabla} f_i^t(\cdot)\|\leq\alpha_t$.
	Letting $z_i(t)=x_i(t+1)-y_i(t)$, by (\ref{step4}), we have
	\begin{equation*}	\label{le2-4}
		\begin{split}
			x_i(t+1)=\sum_{j\in\mathcal{N}_i}a_{ij}(t){x}_{j}(t)+z_i(t)
		\end{split}
	\end{equation*}
	Denote $x(t)=[(x_1(t))^{\top},\cdots,(x_n(t))^{\top}]^{\top}$ and $z(t)=[(z_1(t))^{\top},\cdots,(z_n(t))^{\top}]^{\top}$, one has
	\begin{equation}	\label{le2-5}
		\begin{aligned}
			&x(t+1) \\
			&=(A(t)\otimes I_m)x(t)+z(t)	\\
			&=(A(t:1)\otimes I_m)x(1)+\sum_{s=1}^{t-1}(A(t:s+1)\otimes I_m)z(s)+z(t)
		\end{aligned}
	\end{equation}	
	By the definition of $\bar{x}(t)$, it implies that
	\begin{equation}	\label{le2-6}
		\begin{split}
			\bar{x}(t+1)		&=\frac{1}{n}(1_n^{\top}\otimes I_m)x(t+1)	\\
			&=\frac{1}{n}(1_n^{\top}\otimes I_m)x(1)+\frac{1}{n}\sum_{s=1}^{t}(1_n^{\top}\otimes I_m)z(s).
		\end{split}
	\end{equation}
	Combining (\ref{le2-5}) and (\ref{le2-6}) results in that
	\begin{equation*}	\label{le2-8}
		\begin{split}
			&\|x_i(t+1)-\bar{x}(t+1)\| 	\\
			&\leq\|(([A(t:1)]_i-\frac{1_n^{\top}}{n})\otimes I_m)\|\|x(1)\| 	\\
			&~+\|((e_i^{\top}-\frac{1_n^{\top}}{n})\otimes I_m)\|\|z(t)\|	\\
			&~+\sum_{s=1}^{t-1}\|(([A(t:s+1)]_i-\frac{1_n^{\top}}{n})\otimes I_m)\|\|z(s)\|	\\
			&\leq\frac{\sqrt{nm}\mathcal{C}\lambda^t}{\lambda}\|x(1)\| +\frac{\sqrt{nm}\mathcal{C}}{\mu\lambda}\sum_{s=1}^{t}\alpha_{s}\beta_{s}\lambda^{t-s}
		\end{split}
	\end{equation*}
	where the last inequality holds by using Lemma \ref{le1} and (\ref{le2-3}).
\end{proof}

Based on the lemmas established above, now we present the proof of Theorem 1.

\textbf{Proof of Theorem 1.}
The convexity condition of $f_i^t(\cdot)$ implies
\begin{equation} \label{th1-1}
	\begin{aligned}
		&f_i^t(x_i(t))-f_i^t(x^*(t))	\\
		&\leq\langle\nabla f_i^t(x_i(t)),x_i(t)-x^*(t) \rangle	\\
		&=\langle \nabla f_i^t(x_i(t)),x_i(t)-y_i(t) \rangle+\langle \widetilde{\nabla}f_i^t(x_i(t)),y_i(t)-x^*(t)\rangle \\
		&~+\langle\nabla f_i^t(x_i(t))-\widetilde{\nabla}f_i^t(x_i(t)),y_i(t)-x^*(t) \rangle .	\\
	\end{aligned}
\end{equation}
For the first term of the right-hand side of (\ref{th1-1}), we have
\begin{equation} \label{th1-2}
	\begin{split}
		&\langle\nabla f_i^t(x_i(t)),x_i(t)-y_i(t) \rangle	\\
		&=\langle\nabla f_i^t(x_i(t)),x_i(t)-\bar{x}(t) \rangle+ \langle\nabla f_i^t(x_i(t)),\bar{x}(t)-y_i(t) \rangle	\\
		&=\langle \nabla f_i^t(x_i(t)),x_i(t)-\bar{x}(t) \rangle	\\
		&~+\sum_{j\in\mathcal{N}_i}a_{ij}(t)\langle\nabla f_i^t(x_i(t)),\bar{x}(t)-x_j(t) \rangle	\\
		&\leq G \|x_i(t)-\bar{x}(t)\| +G\sum_{j\in\mathcal{N}_i}a_{ij}(t)\|\bar{x}(t)-x_j(t)\|	\\
	\end{split}
\end{equation}
where the first inequality holds by using the Cauchy–Schwarz inequality.
Letting $x = x^*(t)$ in (\ref{le2-1}) yields
\begin{equation*}	\label{th1-4}
	\begin{split}
		&\langle \beta_{t}\widetilde{\nabla}f_i^t(x_i(t)),x_i(t+1)-x^*(t)\rangle \\
		&\leq\langle\nabla\phi(y_i(t)) -\nabla\phi(x_i(t+1)), x_i(t+1)-x^*(t)\rangle \\
		&=\mathcal{D}_{\phi}(x^*(t),y_i(t))-\mathcal{D}_{\phi}(x^*(t),x_i(t+1))	\\
		&~~-\mathcal{D}_{\phi}(x_i(t+1),y_i(t))	\\
	\end{split}
\end{equation*}
where the first equation results from the definition of $\mathcal{D}_{\phi}(\cdot,\cdot)$.
Hence, for the second term of the right-hand side of (\ref{th1-1}), one has
\begin{equation}	\label{th1-5}
	\begin{split}
		&\langle\beta_t\widetilde{\nabla}f_i^t(x_i(t)),y_i(t)-x^*(t)\rangle \\
		&=\langle\beta_t\widetilde{\nabla}f_i^t(x_i(t)),y_i(t)-x_i(t+1)\rangle	\\
		&~+\langle\beta_t\widetilde{\nabla}f_i^t(x_i(t)),x_i(t+1)-x^*(t)\rangle	\\
		&\leq\frac{\beta_t^2}{2}\|\widetilde{\nabla}f_i^t(x_i(t))\|^2 +\frac{1}{2}\|y_i(t)-x_i(t+1) \|^2	\\
		&~+\langle\beta_t\widetilde{\nabla}f_i^t(x_i(t)),x_i(t+1)-x^*(t)\rangle	\\
		&\leq \frac{\alpha_t^2\beta_t^2}{2} +\mathcal{D}_{\phi}(x^*(t),y_i(t))-\mathcal{D}_{\phi}(x^*(t),x_i(t+1))
	\end{split}
\end{equation}
where the first inequality follows from Young's inequality.
Furthermore,
\begin{equation*} \label{th1-6}
	\begin{split}
		&\sum_{t=1}^{T}\sum_{i=1}^{n}\frac{\mathcal{D}_{\phi}(x^*(t),y_i(t))-\mathcal{D}_{\phi}(x^*(t),x_i(t+1))}{\beta_t}	\\
		&=\sum_{t=1}^{T}\sum_{i=1}^{n}\big(\frac{\mathcal{D}_{\phi}(x^*(t),y_i(t))}{\beta_t} -\frac{\mathcal{D}_{\phi}(x^*(t+1),y_i(t+1))}{\beta_{t+1}}\big)	\\
		&~+\sum_{t=1}^{T}\sum_{i=1}^{n}\frac{\mathcal{D}_{\phi}(x^*(t+1),y_i(t+1))-\mathcal{D}_{\phi}(x^*(t),y_i(t+1))}{\beta_{t+1}}	\\
		&~+\sum_{t=1}^{T}\sum_{i=1}^{n}\frac{\mathcal{D}_{\phi}(x^*(t),y_i(t+1))-\mathcal{D}_{\phi}(x^*(t),x_i(t+1))}{\beta_{t+1}}	\\
	\end{split}
\end{equation*}
\begin{equation*} \label{th1-6}
	\begin{split}
		&~+\sum_{t=1}^{T}\sum_{i=1}^{n}\big(\frac{1}{\beta_{t+1}} -\frac{1}{\beta_{t}}\big)\mathcal{D}_{\phi}(x^*(t),x_i(t+1))	\\
		&\leq\frac{nBL_1}{\beta_1}  +\sum_{t=1}^{T}\sum_{i=1}^{n}\frac{L_1\|x^*(t+1)-x^*(t)\|}{\beta_{t+1}}	\\
		&~+\sum_{t=1}^{T}\sum_{i=1}^{n}\frac{\mathcal{D}_{\phi}(x^*(t),y_i(t+1))-\mathcal{D}_{\phi}(x^*(t),x_i(t+1))}{\beta_{t+1}}	\\
		&~+\big(\frac{1}{\beta_{T+1}}-\frac{1}{\beta_1}\big)nBL_1
	\end{split}
\end{equation*}
where the first inequality hold by Assumptions \ref{as2} and \ref{as5}.
Note that
\begin{equation*}	\label{th1-7}
	\begin{split}
		&\sum_{i=1}^{n}\mathcal{D}_{\phi}(x^*(t),y_i(t+1))-\sum_{i=1}^{n}\mathcal{D}_{\phi}(x^*(t),x_i(t+1))	\\
		&=\sum_{i=1}^{n}\mathcal{D}_{\phi}(x^*(t),\sum_{j=1}^{n}a_{ij}(t)x_j(t+1))	\\
		&~-\sum_{i=1}^{n}\mathcal{D}_{\phi}(x^*(t),x_i(t+1))	\\
		&\leq\sum_{j=1}^{n}\mathcal{D}_{\phi}(x^*(t),x_j(t+1))-\sum_{i=1}^{n}\mathcal{D}_{\phi}(x^*(t),x_i(t+1))	\\
	\end{split}
\end{equation*}
where the first equation holds by using (\ref{step4}).
According to the definition of $\mathcal{R}_i^d(T)$ in (\ref{regret}), we have
\begin{equation} \label{th1-8}
	\begin{split}
		&f^t(x_i(t))-f^t(x^*(t))	\\
		&=\frac{1}{n}\sum_{j=1}^{n}\big( f_j^t(x_i(t)) - f_j^t(\bar{x}(t))\big) \\
		&~+ \frac{1}{n}\sum_{j=1}^{n}\big( f_j^t(\bar{x}(t)) - f_j^t(x_j(t))\big)	\\
		&~+ \frac{1}{n}\sum_{j=1}^{n}\big( f_j^t(x_j(t)) - f_j^t(x^*(t))\big)	\\
		&\leq \frac{1}{n}\sum_{j=1}^{n}G\|x_i(t)-\bar{x}(t)\| +\frac{1}{n}\sum_{j=1}^{n}G\|\bar{x}(t)-x_j(t)\|	\\
		&~+ \frac{1}{n}\sum_{j=1}^{n}\big( f_j^t(x_j(t)) - f_j^t(x^*(t))\big)	\\
	\end{split}
\end{equation}
where the first inequality holds due to the Lipschitz continuity of the function.
Then, substituting (\ref{th1-1})-(\ref{th1-5}) into (\ref{th1-8}) yields
\begin{equation} \label{th-regret}
	\begin{aligned}
		&\mathcal{R}_i^d(T)	\leq\frac{4G}{n}\sum_{t=1}^{T}\sum_{i=1}^{n}\| x_i(t) - \bar{x}(t) \|	+\frac{1}{2}\sum_{t=1}^{T}\alpha_t^2\beta_t	\\
		& ~+\frac{BL_1}{\beta_{T+1}} +\frac{L_1\Xi_T}{\beta_{T+1}}	\\
		&~+\frac{1}{n}\sum_{t=1}^{T}\sum_{i=1}^{n}\langle\nabla f_i^t(x_i(t))-\widetilde{\nabla}f_i^t(x_i(t)),x_i(t+1)-x^*(t) \rangle .\\
	\end{aligned}
\end{equation}
Using Lemmas \ref{le4} and \ref{le2}, inequality (\ref{th-regret}) immediately implies (\ref{th11}).
This completes the proof.

\section{A simulation example} \label{se4}
Consider a network consisting of six sensors, whose goal are to cooperatively estimate a moving target \cite{7399359}. Sensors communicate with their neighbors via a time-varying digraph, as shown in Fig.~\ref{digraph}. Here each sensor only has access to its own function value and the state information of its neighbors.
To achieve the least-squares estimation of the target position, the sensors collaboratively solve the following distributed optimization problem:
\begin{equation*}
	\begin{split}
		&\min_{{x}\in \mathbb{R}}\frac{1}{n}\sum_{i=1}^n f_i^t(x), ~~f_i^t(x)=\frac{1}{2}| y_i(t) - M_i x |^2 		\\
		&\textrm{subject to}~~{x} \in \{x~|~|x|\leq 5\}
	\end{split}
\end{equation*}
where $y_i(t)=M_iz(t)+e_i(t)$ denotes the measurement of sensor $i$,
$M_i$ represents the observation parameter of sensor $i$, $e_i(t)$ represents the adverse noise  of sensor $i$ following an $F$-distribution with a probability density function $f(x;3,5)$, and $z(t)$ represents the target position defined as $z(t)= 0.2z(t-1) + 0.5\cos(t/60) + 0.5$. Here we assume that $M_1=0.5$, $M_2=0.1$, $M_3=2$, $M_4=1$, $M_5=1.2$, $M_6=1.8$.

Algorithm 1 is employed to address the problem.
The step sizes are set as $\alpha_t = 0.2(t+1)^{0.3}+2$, $\beta_{t}=15(t+1)^{-0.6}$, and $\gamma_t=0.2(t+1)^{-0.25}$. The kernel function is defined as $K(r)=\frac{15r}{4}(5-7r^2)$.
By running Algorithm 1 in one round, the trajectories of the target's state and the average state of all sensors are shown in Fig.~\ref{fig2}, where the state of the target is depicted in blue and the average state of all sensors is depicted in orange.
While the bounds of the dynamic regrets are shown in Fig.~\ref{fig3}.
From Fig.~\ref{fig2}, we see that the average state of all sensors approximates to $z(t)$.
Based on Fig.~\ref{fig3}, we can see that $\mathcal{R}_i^d(t)/t$ decays, so $\mathcal{R}_i^d(t)$ grows sublinearly.
The observations are consistent with the results established in Theorem 1. Thus, the effectiveness of Algorithm 1 is further verified.
\begin{figure}
	\centering
	\includegraphics[width=0.4\textwidth]{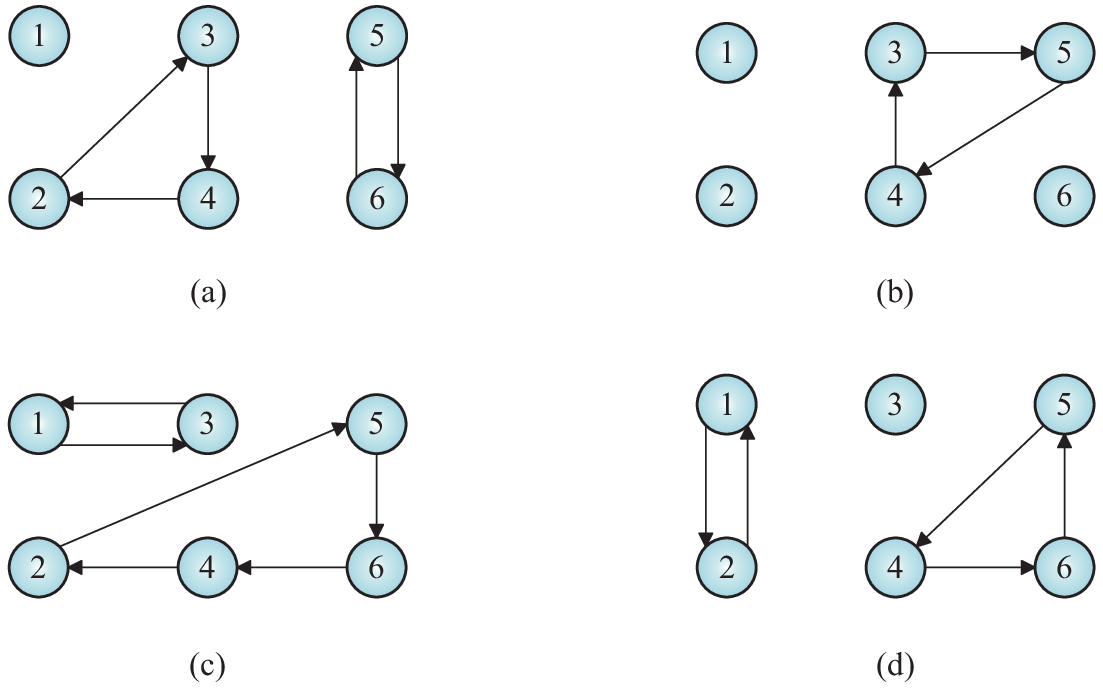}
	\caption{$4$-strongly connected graph. The switching order is given by (a)$\to$(b)$\to$(c)$\to$(d)$\to$(a)$\to\dots$}	\label{digraph}
\end{figure}


\begin{figure}
	\begin{minipage}[t]{0.48\linewidth}
		\centering
		\includegraphics[width=\linewidth]{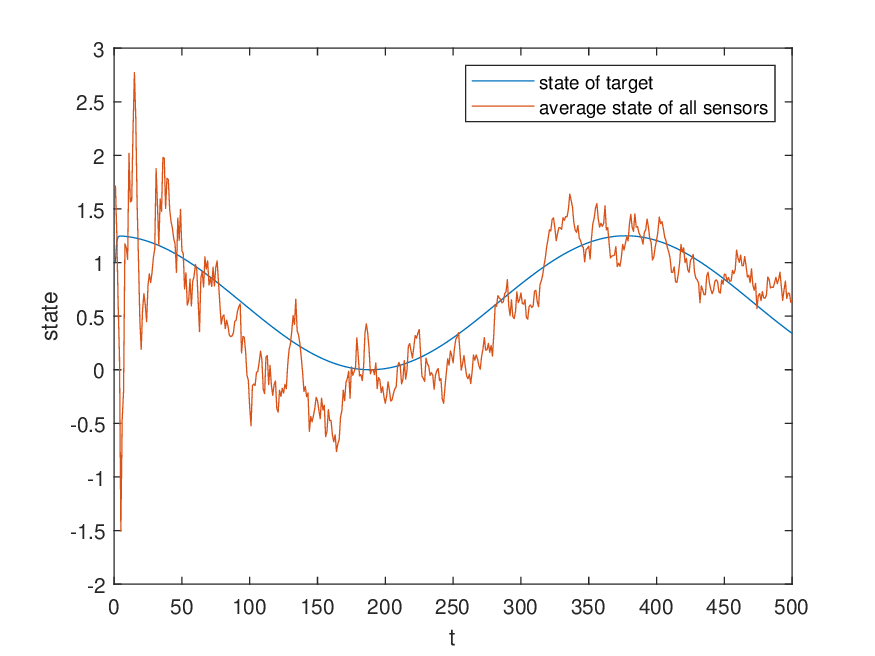}
		\caption{The state of the target and the average state of all sensors.}
		\label{fig2}
	\end{minipage}
	\begin{minipage}[t]{0.48\linewidth}
		\centering
		\includegraphics[width=\linewidth]{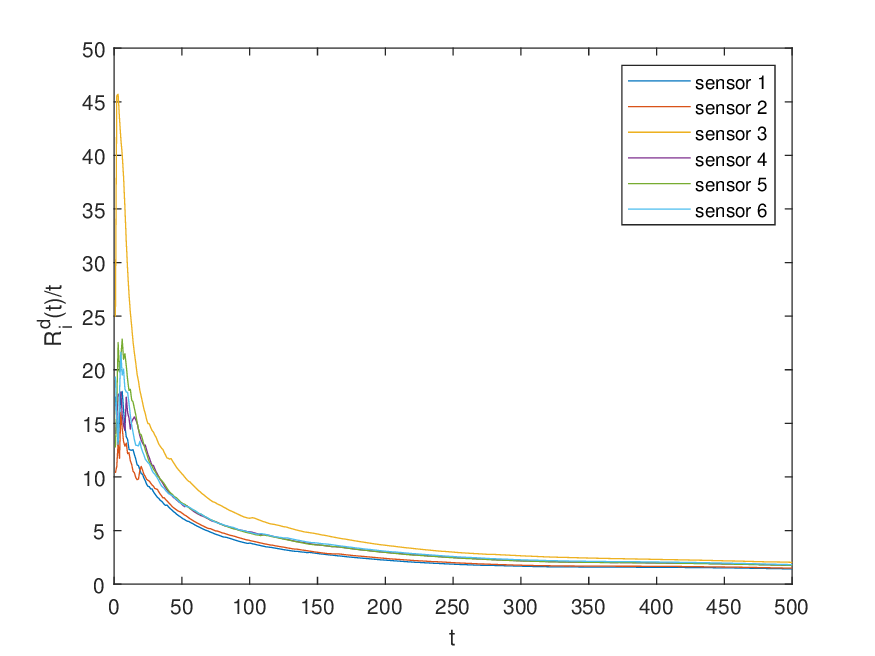}
		\caption{The trajectory of $\mathcal{R}_i^d(t)/t$ under Algorithm 1.}
		\label{fig3}
	\end{minipage}
\end{figure}

\section{Conclusions} \label{se5}
In this paper, the problem of online distributed zeroth-order optimization with non-zero-mean adverse noise has been studied.
Each agent only has access to an estimate of the real gradient by the kernel function-based estimator and exchanges local information with its neighbors via a time-varying digraph.
To address this problem, we propose an online distributed zeroth-order mirror descent algorithm involving the kernel function-based estimator and the clipped strategy.
Under the algorithm, the high probability bound of the dynamic regrets is analyzed.
The results show that, if the graph is uniformly strongly connected and if the variation in the optimal point sequence grows at a certain rate, then the high probability of the dynamic regret increases sublinearly.
In our future work, we will also consider several interesting topics, such as the cases with nonconvex objective functions and inequality constraints, which will bring
new challenges to online distributed zeroth-order optimization with non-zero-mean adverse noises.

\bibliographystyle{unsrt}

\bibliography{cited}

\end{document}